\documentclass[a4paper]{jpconf}
\usepackage{amsmath,amssymb,amsthm}

\newcommand{\bx}{\mathbf{x}}
\newcommand{\hbx}{{\hat{\mathbf{x}}}}
\newcommand{\hby}{{\hat{\mathbf{y}}}}
\newcommand{\hbz}{{\hat{\mathbf{z}}}}
\newcommand{\by}{\mathbf{y}}
\newcommand{\bz}{\mathbf{z}}
\newcommand{\RR}{\mathbb{R}}
\newcommand{\K}{\kappa}
\newcommand{\ds}{\displaystyle}
\newcommand{\bu}{\mathbf{u}}
\newcommand{\bv}{\mathbf{v}}
\def\etab{{\boldsymbol{\xi}}}

\newcommand{\wE}{{\widetilde{\mathbf{E}}}}
\newcommand{\sigmae}{{\sigma_\xi}}

\newtheorem{thm}{Theorem}[section]

\newtheorem{lem}[thm]{Lemma}

\newtheorem{rem}[thm]{Remark}

\begin{document}
\title{Far field imaging of a dielectric inclusion}

\author{Abdul Wahab$^1$, Naveed Ahmed$^2$, Tasawar Abbas$^{3,4}$}

\address{$^1$Department of Mathematics, COMSATS Institute of Information Technology, 47040, Wah Cantt., Pakistan}

\address{$^2$Weierstrass Institute for Applied Analysis and Stochastics, Leibniz Institute in Forschungsverbund Berlin e. V. (WIAS), Mohrenstr. 39, 10117
Berlin, Germany}

\address{$^3$Department of Mathematics, University of Wah, 47040, Wah Cantt., Pakistan}

\address{$^4$Department of Mathematics \& Statistics, FBAS, International Islamic University, 44000, Islamabad, Pakistan}

\ead{Wahab@ciitwah.edu.pk, Naveed.Ahmed@wias-berlin.de, Tasawar44@hotmail.com}

\begin{abstract}
A non-iterative topological sensitivity framework for guaranteed far field detection of a dielectric inclusion is presented. The cases of single and multiple  measurements of the electric far field scattering amplitude at a fixed frequency are taken into account. The performance of the algorithm is analyzed theoretically in terms of resolution, stability, and signal-to-noise ratio.   
\end{abstract}

\section{Introduction}

A thriving interest has been shown in topological sensitivity frameworks to procure solutions of assorted inverse problems especially for detecting small inhomogeneities and cracks embedded in homogeneous media \cite{Princeton, BG, CGGM, DG}. The impetus behind this curiosity is the robustness and simplicity of these algorithms. However, the understanding and utilization was somehow heuristic. For the first time, Ammari et al. rigorously debated these algorithms \cite{TDelastic, AGJK, wahab}. 

Our aim is to design and debate a far field detection algorithm based on topological sensitivity of small dielectric inclusion $D_\rho:=\bz_D+\rho O_D\subset\RR^3$ with position $\bz_D$, scale factor $\rho$ and a smooth bounded domain $O_D\subset\RR^3$.  The inclusion $D_\rho$ (with permittivity $\epsilon_1=\epsilon_0\epsilon_{1r}>0$ and permeability $\mu_0>0$) is assumed to be embedded in $\RR^3$ (with background permittivity $\epsilon_0>0$ and permeability  $\mu_0>0$, setting forth a non-permeable behavior (but the results extend to the case otherwise). Here $\epsilon_{1r}>0$ is the relative permittivity. The medium is probed by the incident electric field 
\begin{eqnarray}
\mathbf{E}_0(\bx):=\nabla_\bx\times(\theta^\perp e^{i\K\theta^T\bx}) = i\K\theta\times \theta^\perp e^{i\K\theta^T\bx}, \qquad \bx\in\RR^3,\label{E0}
\end{eqnarray}
where $\K=\omega\sqrt{\epsilon_0\mu_0}\in\RR_+$ is the wavenumber, $\omega>0$ is the angular frequency,  $\theta\in\mathbb{S}^2:=\{\bx\in\RR^3\,|\,\bx\cdot\bx=1\}$ and $\theta^\perp$ is any vector orthogonal to $\theta$. 
Then, the \emph{total} electric field  $\mathbf{E}_\rho$ in $\RR^3$ in the presence of $D_\rho$ satisfies
\begin{align}
&\nabla\times\nabla\times\mathbf{E}_\rho(\bx)-\K^2
(1-(1-\epsilon_{1r})\mathbf{\chi}_{D_\rho}(\bx))\mathbf{E}_\rho(\bx)= \mathbf{0}, \qquad \bx\in\RR^3,\label{Erho}
\\
&\ds\lim_{|\bx|\to\infty}|\bx|\left[\nabla\times\left(\mathbf{E}_\rho-\mathbf{E}_0\right)\times\widehat{\bx}-i\K\left(\mathbf{E}_\rho-\mathbf{E}_0\right)\right]=\mathbf{0},\label{ErhoSM}
\end{align}
where $\mathbf{\chi}_{D_\rho}(\bx)$ is the characteristic function of $D_\rho$ and $\widehat{\bv}:=\bv/|\bv|$ for all $\bv\in\RR^3\setminus\{\mathbf{0}\}$. Let $
\mathbf{E}^\infty_\rho\in L^2_t(\mathbb{S}^2):=\{\hat\bv\in L^2(\mathbb{S}^2)\,|\,\bv(\hbx)\cdot\hbx=0,\,\forall\,\hbx\in\mathbb{S}^2\}$
(called \emph{far-field amplitude}) be defined by 
\begin{equation}
\mathbf{E}_\rho(\bx)-\mathbf{E}_0(\bx)=\mathbf{E}^\infty_\rho\left(\hat\bx;\theta, \theta^\perp\right){e^{i\K|\bx|}}/{|\bx|}+O\left({1}/{|\bx|^2}\right), \quad \text{as }|\bx|\to\infty.
\end{equation} 
Then,  the following problems are dealt with henceforth.
\begin{enumerate}
\item Given the single data set $\{\mathbf{E}^\infty_\rho(\hat\bx; \theta), \,\forall\, \hat\bx\in\mathbb{S}^2, \theta\in\mathbb{S}^2 \}$, find location $\bz_D$ of inclusion $D_\rho$.

\item Given the multiple data set $\{\mathbf{E}^\infty_\rho(\hat\bx; \theta_j), \,\forall\, \hat\bx\in\mathbb{S}^2, \, (\theta_j)_{j=1}^n\in\mathbb{S}^2, \, n\in \mathbb{N}\}$, find location $\bz_D$.
\end{enumerate}

In Section \ref{Resolution}, a topological sensitivity based location search algorithm is presented and its resolution is ascertained. A stability analysis is performed in Section \ref{Stability}. 

\section{Topological sensitivity framework and resolution analysis}\label{Resolution}

Consider a search point $\bz_S\in\RR^3$ and create a trial inclusion $D_\delta=\bz_S +\delta O_S$ inside $\RR^3$ with permittivity $\epsilon_2=\epsilon_0\epsilon_{2r}>0$, scale factor $\delta>0$ and smooth domain $O_S\subset\RR^3$.
Let $\mathbf{E}_\delta$ be the electric field in $\RR^3$ in the presence of $D_\delta$ subject to incidence $\mathbf{E}_0$ and satisfying  \eqref{Erho}-\eqref{ErhoSM} with $(\epsilon_{1r},\chi_{D_\rho}, D_\rho)$ replaced with $(\epsilon_{2r},\chi_{D_\delta},D_\delta)$.  Let $\mathbf{E}^\infty_\delta$ be the far field amplitude of $\mathbf{E}_\delta$.

Define the misfit $\mathcal{J}[\mathbf{E}_0]$ and its topological derivative (imaging) functional ${\partial}_T\mathcal{J}[\mathbf{E}_0]$ by 
\begin{eqnarray}
\ds\mathcal{J}[\mathbf{E}_0](\bz_S)
:= \frac{1}{2}\int_{\mathbb{S}^2} \left| \mathbf{E}^\infty_{\rho}(\hat\bx;\theta,\theta^\perp)-\mathbf{E}^\infty_{\delta}(\hat\bx;\theta,\theta^\perp)\right|^2 ds(\hbx),
\quad
{\partial}_T\mathcal{J}[\mathbf{E}_0](\bz_S):= -\ds\frac{\partial\mathcal{J}[\mathbf{E_0}](\bz_S)}{\partial(\delta)^3}.\label{Dis}
\end{eqnarray}

It is claimed that $\bz_S\in\RR^3$ relative to which $\partial_T\mathcal{J}[\mathbf{E}_0]$ has a sharp peak is a candidate for $\bz_D$. The claim is justified in the sequel. Recall first the asymptotic expansion of $\mathbf{E}^\infty_{\rho}$ versus $\rho$ \cite{AV}:
\begin{align}
\ds \mathbf{E}^\infty_\rho(\hat\bx;\theta) 
=-{\K^2\rho^3}(4\pi)^{-1}
\left(\epsilon_{1r}^{-1}-1\right)
\left(\mathbf{I}-\hat\bx{\hat\bx}^T\right)
\mathbf{M}_\rho\mathbf{E}_0(\bz_D)
e^{-i\K\hbx^T\bz_D} +O(\rho^4),\label{AE1}
\end{align}
wherein $\mathbf{M}_\rho=\mathbf{M}_\rho(\epsilon_{1r}^{-1},O_D)$ is the polarization tensor. Note that $\mathbf{M}_\rho=3\epsilon_{1r}(\epsilon_{1r}+2)^{-1}|O_D|\mathbf{I}$ for a spherical domain, where $\mathbf{I}$ is the $3\times 3$ identity matrix and superposed $T$ indicates matrix transpose. It is worthwhile precising that $\mathbf{E}^\infty_\delta$ also admits an analogous expansion versus $\delta$. Further, consider the dyadic Green's function  
\begin{eqnarray}
\mathbf{\Gamma}(\bx,\by):=-\epsilon_0\left[\mathbf{I}+\K^{-2}\nabla_\bx\nabla_\bx^T\right] e^{i\K|\bx-\by|}/4\pi|\bx-\by|,
\end{eqnarray}
satisfying
$
\ds\nabla_\bx\times \nabla_\bx \times \mathbf{\Gamma}(\bx,\by)-\K^2 \mathbf{\Gamma}(\bx,\by)= -\epsilon_0\delta_\by(\bx)\mathbf{I}$ for all $\bx,\by\in\RR^3$ 
subject to the \emph{Silver-M\"uller} condition. Here $\delta_{\by}(\cdot)=\delta_0(\cdot-\by)$ is the Dirac mass at $\by$ and $\nabla\times$ acts on matrices column-wise. Moreover, for isotropic materials (see, for example, \cite{AILP})
\begin{eqnarray}
\mathbf{\Gamma}(\bx,\by)=  \mathbf{\Gamma}(\by,\bx) 
\quad\text{and}\quad
\nabla_\bx\times \mathbf{\Gamma}(\bx,\by)=\left[\nabla_\by\times \mathbf{\Gamma}(\by,\bx)\right]^T, \quad \forall \bx,\by\in\RR^3.\label{Pro1}
\end{eqnarray}
Finally, recall that the electric \emph{Herglotz} wave for all $\Phi\in L^2_t(\mathbb{S}^2)$ is defined by (\cite[Ch. 6]{colton}) 
\begin{eqnarray}
\mathcal{H}_E[\Phi](\bz):= \int_{\mathbb{S}^2}\Phi(\hbx)e^{i\K\hbx^T\bz} ds(\hbx), \quad \bz\in\RR^3.
\end{eqnarray}
\begin{lem}\label{lem1} For all search points $\bz_S\in\RR^3$ 
\begin{align}
\partial_T\mathcal{J}[\mathbf{E}_0](\bz_S)
=
-\ds\K^2(4\pi)^{-1}\left(\epsilon_{2r}^{-1}-1\right)\Re e\left\{
\overline{\mathcal{H}_E
[\mathbf{E}^\infty_\rho(\cdot;\theta)](\bz_S)}
\cdot \mathbf{M}_\delta
\mathbf{E}_0(\bz_S)\right\}.
\label{TDeps}
\end{align} 
\end{lem}
\begin{proof}
Expanding misfit \eqref{Dis} and invoking expansion \eqref{AE1} for $\mathbf{E}^\infty_\delta$ render  
\begin{align}
\mathcal{J}&[\mathbf{E}_0](\bz_S) 
-\frac{1}{2}\int_{\mathbb{S}^2}\left|\mathbf{E}^\infty_\rho(\hbx;\theta)\right|^2ds(\hbx)
=
-\Re e\left\{\int_{\mathbb{S}^2}\mathbf{E}^\infty_\delta(\hbx;\theta)\cdot\overline{\mathbf{E}^\infty_\rho(\hbx;\theta)}ds(\hbx)\right\}
+O(\delta^6),
\nonumber
\\
=&
\frac{\K^2\delta^3}{4\pi}\left(\epsilon_{2r}^{-1}-1\right)
\Re e\Big\{\int_{\mathbb{S}^2}\left(\mathbf{I}-\hbx{\hbx}^T\right)
\mathbf{M}_\delta\mathbf{E}_0(\bz_D)
\cdot\overline{\mathbf{E}^\infty_\rho(\hbx;\theta)}e^{-i\K\hbx^T\bz_S} ds(\hbx)
\Big\} +O(\delta^4).
\label{TDexp}
\end{align}
Remark that for any matrix $\mathbf{A}$, vectors $\bu,\bv$,  density $\Phi\in L^2_t(\mathbb{S}^2)$ and $\bx,\bz\in\RR^3$
\begin{align}
&\mathbf{A}\bu\cdot\bv= \mathbf{A}^T\bv\cdot\bu, 
\qquad\qquad 
\left(\mathbf{I}-\hbx\hbx^T\right)\bu=-\hbx\times
(\hbx\times\bu)\label{relation1}
\\
&\K^2\hbx\times\left(\hbx\times \overline{\Phi(\hbx)}\right)e^{-i\K\hbx^T\bz}=- \nabla_{\bz}\times \nabla_{\bz}\times \left(\overline{\Phi(\hbx)e^{i\K\hbx^T\bz}}\right).\label{relation2}
\end{align}
Therefore, by the fact that $\nabla\times\nabla\times \mathcal{H}_E[\Phi]-\K^2\mathcal{H}_E[\Phi]=\mathbf{0}$ for   $\Phi\in L^2_t(\mathbb{S}^2)$ and by above identities
\begin{align*}
\mathcal{J}[\mathbf{E}_0](\bz_S)
&
-\frac{1}{2}\int_{\mathbb{S}^2}\left|\mathbf{E}^\infty_\rho(\hbx;\theta)\right|^2ds(\hbx)
=\ds\frac{\K^2\delta^3}{4\pi}\Big(\epsilon_{2r}^{-1}-1\Big)\Re e\left\{
\overline{\mathcal{H}_E
[\mathbf{E}^\infty_\rho](\bz_S)}
\cdot \mathbf{M}_\delta 
\mathbf{E}_0(\bz_S)\right\}+O(\delta^4).
\end{align*}
Finally, the proof is completed by taking limit as $\delta^3\to 0$.
\end{proof}

In order to substantiate the capability of $\partial_T\mathcal{J}[\mathbf{E}_0]$ to locate $\bz_D$, we elaborate $\mathcal{H}_E[\mathbf{E}^\infty_\rho](\bz)$ as
\begin{align}
\nonumber
\mathcal{H}_E
[\mathbf{E}^\infty_\rho(\cdot;\theta)](\bz)
=&\int_{\mathbb{S}^2} \mathbf{E}^\infty_\rho(\hbx;\theta)e^{i\K\hbx^T\bz}ds(\hbx),
\\
\nonumber
=&
-\frac{\K^2\rho^3}{4\pi}\left(\epsilon_{1r}^{-1}-1\right)\int_{\mathbb{S}^2} \left(\mathbf{I} -\hbx\hbx^T\right)e^{i\K\hbx^T(\bz-\bz_D)}
\mathbf{M}_\rho\mathbf{E}_0(\bz_D)ds(\hbx) +O(\rho^4),
\\
\nonumber
=&
-\frac{\K^2\rho^3}{4\pi}\left(\epsilon_{1r}^{-1}-1\right)\left[\left(\mathbf{I} -\frac{1}{\K^2}\nabla_\bz\nabla_\bz^T\right)\int_{\mathbb{S}^2} e^{i\K\hbx^T(\bz-\bz_D)}ds(\hbx)\right]
\mathbf{M}_\rho\mathbf{E}_0(\bz_D) +O(\rho^4)
\\
=&
{\K\rho^3}\epsilon_0^{-1}\left(\epsilon_{1r}^{-1}-1\right)\Im m\Big\{\mathbf{\Gamma}(\bz,\bz_D)\Big\}\mathbf{M}_\rho\mathbf{E}_0(\bz_D)+ O(\rho^4).\label{HEeps}
\end{align}
Indeed, the last identity is evident since  
\begin{equation}
\int_{\mathbb{S}^2}e^{i\K\hbx^T(\bz-\bz_D)}ds(\hbx)= j_0(\K|\bz-\bz_D|)= {4\pi}{\K^{-1}}\Im m\big\{{e^{i\K|\bz-\bz_D|}}/{4\pi|\bz-\bz_D|}\big\}, 
\end{equation}
where $j_n$ is the order $n$ spherical Bessel function of first kind and consequently, 
\begin{eqnarray}
\left(\mathbf{I}-\frac{1}{\K^2}\nabla_\bz\nabla_\bz^T\right)\int_{\mathbb{S}^2} e^{i\K\hbx^T(\bz-\bz_D)}ds(\hbx)= -\frac{4\pi}{\K\epsilon_0}\Im m\Big\{\mathbf{\Gamma}(\bz,\bz_D)\Big\}.\label{exp1} 
\end{eqnarray}
In fact, the following result is proved by virtue of Lemma \ref{lem1} and expansion \eqref{HEeps}.  
\begin{thm}\label{thm1}
 For all $\bz_S\in\RR^3$ and constant $C_\epsilon:= \left(\epsilon_{1r}^{-1}-1\right)\left(\epsilon_{2r}^{-1}-1\right)$
\begin{align}
\partial_T\mathcal{J}[\mathbf{E}_0](\bz_S)
=
\ds -{\rho^3\K^3C_\epsilon}(4\pi\epsilon_0)^{-1}\Re e\left\{
 \Im m\big\{\mathbf{\Gamma}(\bz_S,\bz_D)\big\}\mathbf{M}_\rho\overline{\mathbf{E}_0(\bz_D)}
\cdot \mathbf{M}_\delta
\mathbf{E}_0(\bz_S)\right\} +O(\rho^4).
\label{TDeps2}
\end{align}  
\end{thm}

\begin{rem} \label{rem1}
It immediately follows from Theorem \ref{thm1} that 
$$
\partial_T\mathcal{J}[\mathbf{E}_0](\bz_S)\propto \Im m\big\{\mathbf{\Gamma}(\bz_S,\bz_D)\big\} =  -\frac{\epsilon_0\K}{4\pi}\left[\frac{2}{3}j_0(\K r)\mathbf{I}+j_2(\K r)\left(\widehat{\mathbf{r}}\,\widehat{\mathbf{r}}^T-\frac{1}{3}\mathbf{I}\right)\right],
$$
where $\mathbf{r}:=\bz_S-\bz_D$ with $r:=|\mathbf{r}|$. Since $j_n(kr)=O(1/kr)$ as $kr\to \infty$ and $j_n(kr)=O((kr)^n)$ as $kr\to 0$, the functional $\bz_S\to \partial_T\mathcal{J}[\mathbf{E}_0](\bz_S)$ rapidly decays for $\bz_S$ away from $\bz_D$ and has a sharp peak when $\bz_S\to\bz_D$ with a focal spot size of half a wavelength of the incident wave.  Moreover, it synthesizes the sensitivity of $\mathcal{J}[\mathbf{E}_0](\bz_S)$ relative to the insertion of an inclusion at $\bz_S\in\Omega$. Heuristically, if the contrast $(\epsilon_{2r}-1)$ has the same sign as $(\epsilon_{1r}-1)$, then  $\mathcal{J}[\mathbf{E}_0](\bz_S)$ must observe the most pronounced decrease at the potential candidate $\bz_S\in\Omega$ for the true location $\bz_D$. In other words, $\bz_S\to \frac{\partial\mathcal{J}[\mathbf{E}_0]}{\partial(\delta^3)}(\bz_S)$ is expected to attain its most pronounced negative value; refer, for instance, to \cite{Bellis, BG} for detailed discussions on sign heuristic. Notice that $C_\epsilon$ is positive if the contrasts of true and trial inclusions have same signs. Consequently, by virtue of the decay property, $\partial_T\mathcal{J}[\mathbf{E}_0](\bz_S)=-\frac{\partial\mathcal{J}[\mathbf{E}_0]}{\partial(\delta^3)}(\bz_S)$
assumes its maximum positive value when $\bz_S\to\bz_D$. In a nutshell, $\partial_T\mathcal{J}[\mathbf{E}_0]$ achieves the detection of $\bz_D$ and resolution limit.
\end{rem}

\subsection{Imaging with multiple measurement}

Let $(\theta_j)_{j=1}^n\in\mathbb{S}^2$ be equidistributed directions and $\theta_j^{\perp,\ell}$ ($\ell=1,2$) be such that $\{\theta_j,\theta_j^{\perp,1},\theta_j^{\perp,2}\}$ is an orthonormal basis of $\RR^3$. We define the incident fields  and the topological sensitivity functional by 
\begin{equation}\label{MultiFunc}
\mathbf{E}_0^{j,\ell}:=i\K\theta_j\times\theta_j^{\perp,\ell}e^{i\K\theta^T_j\bx}
\quad\text{and}\quad
\partial_T\mathcal{J}(\bz_S):= \frac{1}{n}\sum_{\ell=1}^2\sum_{j=1}^n\partial_T\mathcal{J}[\mathbf{E}_0^{j,\ell}](\bz_S). 
\end{equation}
As for $n$ sufficiently large $\frac{1}{n}\sum_{j=1}^n e^{i\K\theta_j^T(\bx-\by)}\approx j_0(\K|\bx-\by|)$ and $\big\{\theta_j,\theta_j^{\perp,1},\theta_j^{\perp,2}\big\}$ is a basis of $\RR^3$, therefore
\begin{align}
\nonumber
\frac{1}{n}\sum_{\ell=1}^2\sum_{j=1}^n e^{i\K\theta_j^T(\bx-\by)} &
\theta_j^{\perp,\ell} \left(\theta_j^{\perp,\ell}\right)^T
= \frac{1}{n}\sum_{j=1}^n \left(\mathbf{I}-\theta_j\theta_j^T\right)
 e^{i\K\theta_j^T(\bx-\by)}
\\
\nonumber
=& \frac{1}{n}\sum_{j=1}^n\left(\mathbf{I}-\frac{1}{\K^2}\nabla_\bx\nabla_\bx^T\right) e^{i\K\theta_j^T(\bx-\by)}
\approx
-{4\pi}{(\K\epsilon_0)}^{-1}\Im m\big\{\mathbf{\Gamma}(\bx,\by)\big\}.\label{theta1}
\end{align}
Similarly,  
\begin{equation}
\frac{1}{n}\sum_{\ell=1}^2\sum_{j=1}^n  e^{i\K\theta_j^T(\bx-\by)} \left(\theta_j\times \theta_j^{\perp,\ell}\right)\left(\theta_j\times \theta_j^{\perp,\ell}\right)^T
\approx
-{4\pi}{(\K\epsilon_0)}^{-1}\Im m\big\{\mathbf{\Gamma}(\bx,\by)\big\}.\label{theta2}
\end{equation}
Following result holds.
\begin{thm}\label{thm2}
Let $\bz_S\in\RR^3$, $n\in\mathbb{N}$ be sufficiently large and $\mathbf{A}:\mathbf{B}=\sum_{i,j=1}^3 \mathbf{A}_{ij}\mathbf{B}_{ij}$. Then
\begin{align}
\partial_T\mathcal{J}(\bz_S)
\approx {\rho^3\K^4}\epsilon_0^{-2} C_\epsilon\Re e\Big\{\Im m\left\{\mathbf{\Gamma}(\bz_S,\bz_D)\right\} \mathbf{M}_{\rho}: \mathbf{M}_{\delta}\Im m\left\{\mathbf{\Gamma}(\bz_S,\bz_D)\right\} \Big\}+ O(\rho^4).
\label{TDmulti2}
\end{align}
\end{thm}
\begin{proof}
Since $\mathbf{A}\theta\cdot\theta =\mathbf{A}: \theta\theta^T$, the approximation \eqref{theta2} yields
\begin{align*}
&\partial_T\mathcal{J}(\bz_S)=
-\frac{\rho^3\K^3 C_\epsilon}{4\pi\epsilon_0 n}\sum_{\ell=1}^2\sum_{j=1}^n
\Re e\Big\{
 \Im m\left\{\mathbf{\Gamma}(\bz_S,\bz_D)\right\}\mathbf{M}_\rho
\overline{\mathbf{E}_0^{j,\ell}(\bz_D)}
\cdot \mathbf{M}_\delta\mathbf{E}_0^{j,\ell}(\bz_S)\Big\} +O(\rho^4),
\\
=& 
- \frac{\rho^3\K^5C_\epsilon}{4\pi\epsilon_0 n}\sum_{\ell=1}^2\sum_{j=1}^n
\Re e\Big\{
\Im m\left\{\mathbf{\Gamma}(\bz_S,\bz_D)\right\}\mathbf{M}_\rho
:\mathbf{M}_\delta(\theta_j\times\theta_j^{\perp,\ell})(\theta_j\times\theta_j^{\perp,\ell})^Te^{i\K\theta_j^T(\bz_S-\bz_D)}
\Big\} +O(\rho^4),
\\
\approx &\phantom{-}\, 
{\rho^3\K^4}\epsilon_0^{-2}C_\epsilon
\Re e\big\{
\Im m\left\{\mathbf{\Gamma}(\bz_S,\bz_D)\right\}\mathbf{M}_\rho
:\mathbf{M}_\delta\Im m\left\{\mathbf{\Gamma}(\bz_S,\bz_D)\right\}
\big\}+O(\rho^4).
\end{align*}
\end{proof}
Theorem \ref{thm2} substantiates that $\partial_T\mathcal{J}(\bz_S)\propto \Im m\left\{\mathbf{\Gamma}(\bz_S,\bz_D)\right\} $ and has a peak (sharper than that of $\partial_T\mathcal{J}[\mathbf{E}_0]$) when $\bz_S\to\bz_D$ (see Remark \ref{rem1}). 

\section{Stability with respect to measurement noise}\label{Stability}

The aim in this section is to establish stability of the imaging functional $\partial_T\mathcal{J}$. Assume that the measurements of the far field amplitude are corrupted by a mean-zero circular Gaussian noise $\etab:\mathbb{S}^2\to\mathbb{C}^3$, that is, 
$\mathbf{E}_\rho^\infty(\hbx):= \wE^{\infty}_\rho(\hbx)+\etab(\hbx)$ for all $\hbx\in\mathbb{S}^2$, where $\mathbf{E}_\rho^\infty$ is the corrupted value. Henceforth, a superposed $\sim$ indicates true value without noise corruption.  Further, let $\sigmae$ be the noise covariance of $\etab$ such that for all $\hby,\hby'\in\mathbb{S}^2$
\begin{eqnarray}
\mathbb{E}\left[\etab(\hby)\overline{\etab(\hby')}^T\right]
=\sigmae^2\delta_\hby(\hby')\mathbf{I}
\qquad\text{and}\qquad
\mathbb{E}\left[\etab^j(\hby)\overline{\etab^{j'}(\hby')}^T\right]
=\sigmae^2\delta_{jj'}\delta_\hby(\hby')\mathbf{I},\label{EEcor}
\end{eqnarray}
where $j$ and $j'$ indicate the $j$-th and $j'$-th measurements and $\delta_{jj'}$ is Kronecker's delta. 

The noise affects $\partial_T\mathcal{J}[\mathbf{E}_0]$ and $\partial_T\mathcal{J}$ by means of Herglotz wave, $\mathcal{H}_E[\mathbf{E}^\infty_\rho]=\mathcal{H}_E[\wE^\infty_\rho]+   \mathcal{H}_E[\etab]$. The term $\mathcal{H}_E[\wE^\infty_\rho]$ is independent of noise and yields true image; discussed earlier. On the other hand, $\mathcal{H}_E[\etab]$ is a circular Gaussian random process with mean-zero by linearity and by definition of $\etab$. Furthermore, for all $\hbz,\hbz'\in\mathbb{S}^2$ 
\begin{align}
\mathbb{E}\Big[\mathcal{H}_E[\etab](\bz)
\overline{\mathcal{H}_E[\etab](\bz')}^T\Big]
=&\mathbb{E}\left[\iint_{\mathbb{S}^2\times\mathbb{S}^2}
\etab(\hbx)\overline{\etab(\hby)^T} e^{i\K\hbx^T\bz} e^{-i\K\hby^T\bz'} ds(\hbx)ds(\hby)\right],
\nonumber
\\
=&\mathbb{E}\left[\iint_{\mathbb{S}^2\times\mathbb{S}^2}
\hbx\times(\etab(\hbx)\times\hbx)\overline{\etab(\hby)^T} e^{i\K\hbx^T\bz} e^{-i\K\hby^T\bz'} ds(\hbx)ds(\hby)\right].
\nonumber
\end{align}
To get the latter identity, the fact that $\etab\in L^2_T(\mathbb{S}^2)$ is used. By virtue of \eqref{relation1}-\eqref{relation2}, one gets
\begin{align}
\mathbb{E}\Big[\mathcal{H}_E[\etab](\bz)
\overline{\mathcal{H}_E[\etab](\bz')}^T\Big]
=&\mathbb{E}\left[\iint_{\mathbb{S}^2\times\mathbb{S}^2}
(\mathbf{I}-\hbx\hbx^T)\etab(\hbx)\overline{\etab(\hby)^T} e^{i\K\hbx^T\bz} e^{-i\K\hby^T\bz'} ds(\hbx)ds(\hby)\right],
\nonumber
\\
=&\sigmae^2\left(\mathbf{I}-\frac{1}{\K^2}\nabla_\bz\nabla_\bz^T\right) \int_{\mathbb{S}^2}e^{i\K\hbx^T(\bz-\bz')}ds(\hbx)  =-\frac{4\pi\sigmae^2}{\K\epsilon_0}\Im m\big\{\mathbf{\Gamma}(\bz,\bz')\big\} .\label{Cov1}
\end{align}
Thus, $\mathcal{H}_E[\etab]$ precipitates a speckle pattern, that is, a random cloud of hot spots having typical diameters of the order of wavelength and amplitudes of the order of $\sigmae /\sqrt{\K}$.  The covariance of corrupted image is furnished by Theorem \ref{thm1} and expression \eqref{Cov1} as
\begin{align*}
{\rm Cov}
\big(\partial_T\mathcal{J}(\bz),&\partial_T\mathcal{J}(\bz')\big)
\\
=&
\frac{a_\epsilon^2}{n^2}\sum_{\ell,\ell'=1}^2\sum_{j,j'=1}^n
\mathbb{E}\Big[
\ds\Re e\left\{
\overline{\mathcal{H}_E
[\etab^{j,\ell}](\bz)}
\cdot \mathbf{M}_\delta
\mathbf{E}_0^{j,\ell}(\bz)\right\}
\Re e\left\{
\overline{\mathcal{H}_E
[\etab^{j',\ell'}](\bz')}
\cdot \mathbf{M}_\delta
\mathbf{E}_0^{j',\ell'}(\bz')\right\}
\Big],
\\
=
&
\frac{a_\epsilon^2}{2n^2}\sum_{\ell=1}^2\sum_{j=1}^n
\ds\Re e\left\{\mathbf{M}_\delta
\mathbf{E}_0^{j,\ell}(\bz)\cdot 
\mathbb{E}\Big[
\overline{\mathcal{H}_E
[\etab^{j,\ell}](\bz)}
\Big(\mathcal{H}_E [\etab^{j,\ell}](\bz')\Big)^T\Big]
\mathbf{M}_\delta
\overline{\mathbf{E}_0^{j,\ell}(\bz')}\right\},
\\
=
&
-\frac{2\pi a_\epsilon^2\sigmae^2}{n^2\K\epsilon_0}\sum_{j=1}^n
\ds\Re e\left\{\mathbf{M}_\delta
\mathbf{E}_0^{j,\ell}(\bz)\cdot 
\Im m\big\{\mathbf{\Gamma}(\bz,\bz')\big\}
\mathbf{M}_\delta
\overline{\mathbf{E}_0^{j,\ell}(\bz')}\right\},
\end{align*}
where  $a_\epsilon= \K^2(4\pi)^{-1} \left(\epsilon_{2r}^{-1}-1\right)$ and the identities \eqref{EEcor}-\eqref{Cov1} are invoked.

Now substituting expression for $\mathbf{E}^{j,\ell}_0$ and using approximation \eqref{theta2}, we arrive at
\begin{align*}
{\rm Cov}\big(\partial_T&\mathcal{J}(\bz),\partial_T\mathcal{J}(\bz')\big)
\\
=
&
-\frac{2\pi a_\epsilon^2\sigmae^2\K}{n^2\epsilon_0}\sum_{\ell=1}^2\sum_{j=1}^n
\ds\Re e\left\{\mathbf{M}_\delta(\theta_j\times\theta_j^{\perp,\ell})\cdot 
\Im m\big\{\mathbf{\Gamma}(\bz,\bz')\big\}
\mathbf{M}_\delta
(\theta_j\times\theta_j^{\perp,\ell}) e^{i\K\theta^T(\bz-\bz')}\right\},
\\
\approx
&
{\K^4\sigmae^2(\epsilon_{2r}^{-1}-1)^2}(2n)^{-1}\epsilon_0^{-2}
\ds\Re e\Big\{\mathbf{M}_\delta \Im m\big\{\mathbf{\Gamma}(\bz,\bz')\big\}: 
\Im m\big\{\mathbf{\Gamma}(\bz,\bz')\big\}
\mathbf{M}_\delta
\Big\},
\end{align*}
In order to fathom the statistics of corrupted image, assume for an instance that $D_\delta$ is a sphere so that $\mathbf{M}_\delta=3(2\epsilon_{2r}^{-1}+1)^{-1}|O_S|\mathbf{I}$. Then, it can be readily verified that
\begin{align}
{\rm Cov}\big(\partial_T\mathcal{J}(\bz),\partial_T\mathcal{J}(\bz')\big)
\approx
\ds{b_\epsilon^2\sigmae^2\K^4}{(2n)^{-1}}\left\|\Im m\big\{\mathbf{\Gamma}(\bz,\bz')\big\}\right\|^2,\label{Cov2}
\end{align}
where  ${b}_\epsilon= 3|O_S|(\epsilon_{2r}^{-1}-1)/\epsilon_0(2\epsilon_{2r}^{-1}+1) $ and  $\|\mathbf{A}\|:= \sqrt{\mathbf{A}:\mathbf{A}}$ is the \emph{Frobenius} norm of $\mathbf{A}$. 
\begin{rem}\label{rem3}
The above analysis elucidates that noise perturbation of image is typically of order $\sigmae/\sqrt{2n}$ and the hot spots have shape identical with that of actual peak procured by $\partial_T\mathcal{J}$ subject to accurate measurements. The main spike of $\partial_T\mathcal{J}$ is not altered. Further, the dependence of typical perturbation size on  $1/\sqrt{2n}$ indicates that the topological sensitivity based imaging framework is more stable when  multiple measurements are available at hand. 
\end{rem}

\subsection{Signal-to-noise ratio}

It follows immediately from \eqref{Cov2} that the variance of  $\partial_T\mathcal{J}$ at $\bz\in\RR^3$ can be approximated by 
\begin{eqnarray}
\label{var1} 
{\rm Var}\big(\partial_T\mathcal{J}(\bz)\big)\approx {b_\epsilon^2\sigmae^2\K^4}{(2n)^{-1}}
\big\|\Im m\big\{\mathbf{\Gamma}(\bz,\bz)\big\}
\big\|^2.
\end{eqnarray}
Specifically when $D_\rho$ is spherical inclusion, the signal-to-noise ratio (SNR) defined by 
$
{\rm SNR}= \mathbb{E}\left[\partial_T\mathcal{J}(\bz_D)\right]\left({\rm Var}\left[\partial_T\mathcal{J}(\bz_D)\right]\right)^{-1/2},
$
can be approximated by virtue of \eqref{var1} and Theorem \ref{thm2} as
\begin{eqnarray}
{\rm SNR}\approx 3\sqrt{2n}\sigmae^{-1}\epsilon_0^{-1}|2\epsilon_{1r}^{-1}+1|^{-1}|\epsilon_{1r}^{-1}-1|\rho^3|O_D|\K^2\big\|\Im m\big\{\mathbf{\Gamma}(\bz_D,\bz_D)\big\}\big\|.\label{SNR}
\end{eqnarray}
The expression \eqref{SNR} indicates that signal-to-noise ratio depends directly on contrast, volume of the inclusion, the operating frequency and inversely on noise covariance.
\section{Conclusion}

The topological sensitivity functions presented herein for detecting a dielectric inclusion, from single or multiple measurements of electric far field amplitude, are proved to be very efficient in terms of resolution, stability and signal-to-noise ratio. In fact, the Rayleigh's resolution limit is achieved and stability with respect to measurement noise is substantiated. The stability with respect to medium noise is not dealt with, however,  moderate stability can be guaranteed under Born approximation. The reader is referred, for instance, to \cite{TDelastic, wahab} in this regard.

\end{document}